\documentclass[
]{ceurart}

\usepackage{listings}
\usepackage[ruled]{algorithm}
\usepackage[noend]{algpseudocode}
\usepackage{amsmath,amssymb}

\newtheorem{theorem}{Theorem}[section]
\newtheorem{proposition}[theorem]{Proposition}

\newdefinition{definition}[theorem]{Definition}
\newproof{proof}{Proof}

\usepackage{subcaption}

\newdefinition{example}[theorem]{Example}
\newdefinition{remark}[theorem]{Remark}

\begin{document}

\copyrightyear{2022}
\copyrightclause{Copyright for this paper by its authors.
  Use permitted under Creative Commons License Attribution 4.0
  International (CC BY 4.0).}

\conference{}

\title{Extending CDCL-based Model Enumeration with Weights}


\author[1]{Giuseppe Spallitta}[%
orcid=0000-0002-4321-4995,
email=gs81@rice.edu,
]
\cormark[1]
\address[1]{Rice University, Houston, Texas, United States of America}

\author[1]{Moshe Vardi}[%
orcid=0000-0002-0661-5773,
email=vardi@rice.edu,
]


\cortext[1]{Corresponding author.}

\begin{abstract}
  In this work we investigate \emph{Weighted Model Enumeration} (WME): given a Boolean formula and a weight function over its satisfying assignments, enumerate models while accounting for their weights. This setting supports weight-driven queries, such as producing the top-$k$ models or all models above a threshold. While related to AllSAT, Weighted Model Counting, and MaxSAT, these paradigms do not treat selective enumeration under weights as a native solver task. We present CDCL-based algorithms for WME that integrate weight propagation, weight-based pruning, and \emph{weight-aware} conflict analysis into both chronological and non-chronological backtracking frameworks. Chronological backtracking exploits implicit blocking to keep the clause database compact, thereby reducing the memory footprint and enabling efficient propagation. In contrast, non-chronological backtracking with clause learning supports explicit blocking and restarts. We show that both approaches are feasible and complementary, highlighting trade-offs in pruning effectiveness with weights and clarifying when each performs best. 
\end{abstract}


\maketitle

\section{Introduction}

In many artificial-intelligence applications, including probabilistic inference \cite{li2011exploiting}, structured decision-making \cite{liu2012belief}, fault diagnosis \cite{cai2017bayesian}, and interpretable machine learning \cite{rudin2022interpretable}, the objective is not merely to obtain a satisfying assignment of a Boolean formula, but to identify a set of "high-quality" models. Such models often correspond to the most probable explanations under a weight function. Classical formulations such as Maximum a Posteriori (MAP) \cite{tolpin2015maximum} and Most Probable Explanation (MPE) \cite{sang2007dynamic,  ChanD12} capture this idea by seeking the single highest-scoring assignment. Yet in practice, many applications require richer outputs: multiple explanations, such as top-$k$ weighted models or all models above a weight threshold, are indispensable for robust decision support and interpretability. This need has long been recognized in database theory through top-$k$ query processing \cite{lawler1972procedure,fagin2001optimal} and in probabilistic inference through top-$k$ explanation tasks \cite{yanover2003finding}.


 Existing CDCL-based reasoning tools fall short of providing native support for weight-aware enumeration of Boolean formulae. AllSAT solvers \cite{toda2016implementing, liang2022allsatcc, gebser2007clasp, spallitta2024disjoint} enumerate indiscriminately, ignoring weights entirely. Weighted Model Counting (WMC) \cite{sang2004combining, oztok2015top, lagniez2017improved, dudek2020dpmc, dilkas2021conditionalwmc} computes global aggregates but provides no access to individual weighted models. MaxSAT solvers \cite{davies2011solving, nadel2024ttopenwboinc, openwboinc2} return a single maximum-weight solution, but lack mechanisms for ranking or filtering multiple alternatives. The 2020 MaxSAT Evaluation introduced a dedicated track for top-$k$ enumeration \cite{bacchus2020maxsat}, highlighting both the recognized importance and the unsolved nature of top-$k$ reasoning. Thus, none of these approaches natively supports selective enumeration under weight constraints. As a result, a natural and practically important task remains unexplored: \emph{Weighted Model Enumeration} (WME).


This work formalizes and develops a CDCL-native framework for WME. Although the problem is naturally connected to both AllSAT and MaxSAT, addressing it efficiently requires solver mechanisms that go beyond a straightforward combination of the two. Beyond the definition, we introduce weight-aware propagation and pruning rules that enable the solver to prune infeasible branches early, and weight-conflict clauses that integrate weight reasoning into CDCL's clause-learning algorithms. These mechanisms reshape the solver's behavior, enabling both threshold-based and top-$k$ enumeration to be performed natively. 

A central issue, then, is how weight conflicts are resolved. In classical SAT solving, non-chronological backtracking and state-of-the-art CDCL techniques are considered indispensable for efficiently scanning the search space \cite{marques2002grasp}, and their benefits have been discussed for enumeration \cite{toda2016implementing}. Yet their advantages may not carry over when weight thresholds guide the search. Conversely, chronological backtracking, recently shown to be effective for model enumeration \cite{spallitta2024disjoint, spallitta2025disjoint}, could align more naturally with WME.
This duality raises the question of which backtracking approach is better suited for weight-aware enumeration, motivating the development and empirical comparison of both chronological and non-chronological CDCL frameworks for WME.

\subsection{Contributions}
\begin{itemize}
\item  We formally introduce \emph{Weighted Model Enumeration}, a generalization of classical reasoning problems, such as AllSAT, WMC, and MaxSAT,  supporting the enumeration of weighted satisfying assignments, possibly under weight-based criteria.
\item We introduce weight-aware propagation and pruning within CDCL-based enumeration, enabling early detection of infeasible branches.
\item  We develop two CDCL frameworks for WME: a \emph{chronological} variant (implicit blocking clauses, low memory footprint, order-sensitive) and a \emph{non-chronological} variant (explicit blocking clauses, restart-friendly).
\item We implement both designs and study their behavior across thresholds and weight distributions, highlighting regimes where one approach outperforms the other.
\end{itemize}

\subsection{Related Work}

\textbf{AllSAT} solvers such as \texttt{clasp}~\cite{gebser2007clasp}, the set of tools \texttt{BC}, \texttt{NBC} and \texttt{BDD}, all three presented in \cite{toda2016implementing}, and, more recently, \texttt{AllSATCC} \cite{liang2022allsatcc} and \texttt{TabularAllSAT}~\cite{spallitta2024disjoint,spallitta2025disjoint} enumerate \emph{all} satisfying assignments without regard to weight. Introducing weighted and filtered results based on weight constraints must be done post hoc, which is computationally prohibitive in large solution spaces.

\textbf{WMC} tools such as \texttt{Ganak} \cite{SRSM19, SM2025}, \texttt{D4} \cite{lagniez2017improved}, \texttt{Cachet} \cite{sang2004combining}, \texttt{miniC2D} \cite{oztok2015top}, \texttt{ADDMC}~\cite{dudek2019addmc}, \texttt{DPMC}~\cite{dudek2020dpmc}, \texttt{sharpSAT-TD} \cite{korhonen2023sharpsat} and algebraic extensions~\cite{dilkas2021conditionalwmc} compute aggregate weights over all models. These are central in probabilistic inference but do not enumerate or rank individual assignments.


\textbf{MaxSAT} solvers including \texttt{MaxHS}~\cite{davies2011solving}, \texttt{Open-WBO}~\cite{martins2014openwbo}, {\tt MPE-SAT} \cite{sang2007dynamic}, and the newer \texttt{Open-WBO-Inc}~\cite{openwboinc2} and \texttt{TT-Open-WBO-Inc}~\cite{nadel2024ttopenwboinc} optimize for a single maximum-weight solution. While multiple models can be obtained via repeated MaxSAT calls, this is not a native ranked-enumeration mechanism and does not scale well.
The MaxSAT Evaluation's experimental top-$k$ track~\cite{bacchus2020maxsat} further highlights that ranked extraction is not a standard MaxSAT objective, as it primarily evaluates correctness (i.e., number and ordering of solutions, possibly with approximation) rather than efficiency. Leading approaches (e.g., RC2 \cite{DBLP:journals/jsat/IgnatievMM19} and MaxHS \cite{davies2011solving}) rely on iterative re-optimization with blocking clauses rather than integrated enumeration and often limit preprocessing, which hinders scalability when extracting multiple solutions.



A possible approach to target WME is to compile the formula into a tractable representation from knowledge compilation~\cite{darwiche2002knowledge}
(e.g., d-DNNF and related decision-diagram formalisms) and then enumerate solutions in ranked order. Recent work started investigating the idea of ranked
enumeration on circuits under restricted scoring schemes~\cite{AmarilliBCM24, MINATO2025467}. This direction could be effective when compilation is feasible and results in a compact representation.

Our focus in this manuscript, however, is different: we target anytime, solver-native top-$k$/threshold enumeration directly within CDCL-style search. This is already a natural fit in the propositional setting, since the weight bounds can be maintained online during the search and used immediately for pruning. The choice is also motivated by how naturally it carries over to SMT. In that setting, a purely propositional compiled representation is not enough: an assignment may satisfy the Boolean abstraction while still being theory-inconsistent. Therefore, a compilation-based approach would need to account not only for propositional structure, but also for the effect of theory reasoning and theory-pruned branches, which is technically nontrivial~\cite{michelutti2024canonical}. By contrast, CDCL($\mathcal{T}$)~\cite{nieuwenhuis2006solving} already combines Boolean search with online theory consistency checking, making the integration of weight reasoning into the search loop more straightforward.

An alternative related line of work is the dynamic-programming optimization tool {\sc DPO} \cite{phan2022dpo}. {\sc DPO} targets optimization on hybrid constraint systems (including XORs) by combining SAT reasoning with dynamic programming over the input problem structure. While {\sc DPO} supports computing an optimal high-weight model, its optimization framework does not allow a trivial modification to enumerate multiple high-weight satisfying models.

\section{Weighted Model Enumeration}

Let $F$ be a CNF Boolean formula over variables $\mathcal{V}=\{A_1,\dots,A_n\}$. A total assignment $\eta:\mathcal{V}\to\{\bot,\top\}$ satisfies $F$ if all clauses of $F$ evaluate to true under $\eta$. A partial assignment $\mu$ assigns only a subset of the variables in $\mathcal{V}$. We represent assignments as sets of literals, e.g.\ $\mu=\{A_1,\neg A_2,A_3\}$, or as conjunctions of literals, e.g.\
$\mu=(A_1\wedge \neg A_2\wedge A_3)$, using the two representations
interchangeably.
In the CDCL loop, $\mu$ denotes the current trail, i.e., an ordered partial assignment partitioned into decision levels. Let $\varepsilon$ denote the empty trail.
We define a \emph{weight function} $w:\mathcal{L}\to\mathbb{Q}_{> 0}$ over the literal set $\mathcal{L}=\{A_i,\neg A_i \mid A_i\in\mathcal{V}\}$.
Finally, let $\mathrm{vars}(\mu)=\{\mathrm{var}(\ell)\mid \ell\in\mu\}$, where $\mathrm{var}(\ell)$ denotes the Boolean variable underlying literal~$\ell$.

\begin{definition}[Weighted Model Enumeration]
    Given a Boolean formula $F$ over variables $\mathcal{V}$ and a literal-level weight function $w: \mathcal{L} \to \mathbb{Q}_{> 0}$, the \emph{Weighted Model Enumeration (WME)} problem asks to compute the set of weighted satisfying assignments:
\[
\mathcal{M}_w(F) = \{ (\eta, w(\eta)) \mid \eta \models F \} \quad\text{where }\quad w(\eta) \;=\; \prod_{\ell \in \eta} w(\ell).
\]
\end{definition}
The goal is to enumerate every model $\eta$ satisfying $F$ together with its aggregated weight $w(\eta)$. The weighted model enumeration problem generalizes:
\begin{itemize}
    \item \textbf{AllSAT}: when $w(\ell) = 1$ for all $\ell$, we get the unweighted set of satisfying assignments;
    \item \textbf{WMC}: when only the scalar quantity $\sum_{\eta \models F} w(\eta)$ is returned, rather than the full set of weighted models.
\end{itemize}

To illustrate the WME task, consider a simple weighted CNF instance over Boolean variables $A_1$, $A_2$, and $A_3$:
\begin{center}
\begin{minipage}[c]{0.42\textwidth}
\[
F = (A_1 \lor A_2) \land A_3
\]
\end{minipage}
\hfill
\begin{minipage}[c]{0.52\textwidth}
\centering
\begin{tabular}{c|ccc}
\toprule
 & $A_1$ & $A_2$ & $A_3$ \\ \midrule
$w(\ell)$ & 0.6 & 0.8 & 0.5 \\
$w(\neg \ell)$ & 0.4 & 0.2 & 0.5 \\ 
\bottomrule
\end{tabular}
\end{minipage}
\end{center}

Table~\ref{tab:wme_example} lists all $8$ possible assignments, indicating which ones satisfy $F$, and reports their corresponding weights for those that do.
\begin{table}[t]
\centering
\begin{tabular}{ccc|c|c|@{\,\qquad\,}|ccc|c|c}
\toprule
\textbf{$A_1$} & \textbf{$A_2$} & \textbf{$A_3$} & \textbf{Sat.\ $F$} & \textbf{Weight} &
\textbf{$A_1$} & \textbf{$A_2$} & \textbf{$A_3$} & \textbf{Sat.\ $F$} & \textbf{Weight} \\
\midrule
0 & 0 & 0 & \texttimes & -- &
1 & 0 & 0 & \texttimes & -- \\
0 & 0 & 1 & \texttimes & -- &
1 & 0 & 1 & \checkmark & $0.6 \cdot 0.2 \cdot 0.5 = 0.06$ \\
0 & 1 & 0 & \texttimes & -- &
1 & 1 & 0 & \texttimes & -- \\
0 & 1 & 1 & \checkmark & $0.4 \cdot 0.8 \cdot 0.5 = 0.16$ &
1 & 1 & 1 & \checkmark & $0.6 \cdot 0.8 \cdot 0.5 = 0.24$ \\
\bottomrule
\end{tabular}
\caption{All possible assignments of $F = (A_1 \lor A_2) \land A_3$ and their associated weights.}
\label{tab:wme_example}
\end{table}

In addition to the baseline WME task, we define two constrained variants:
\begin{itemize}
    \item \textbf{Threshold WME}: Given a threshold $\theta$, enumerate all models $\eta \models F$ such that $w(\eta) \geq \theta$\footnote{We focus on lower-bound thresholds $w(\eta)\ge\theta$ in this paper. Other relational operators (e.g., $>$, $\le$, $<$) can be handled analogously by reversing the comparison and/or using corresponding upper/lower residual bounds.}.
    \item \textbf{Top-$k$ WME}: Enumerate the $k$ models with the highest weights among all models in $F$.
\end{itemize}

These two modes are general enough to express a wide range of practical requirements, including top-$k$ explanation, utility-based model extraction, and probabilistic inference. 

\begin{remark}
Throughout this work, we adopt the multiplicative aggregation of weights,
\[
w(\eta)=\prod_{\ell\in\eta} w(\ell),
\]
as commonly used in Weighted Model Counting and probabilistic reasoning.
While the examples in this paper use weights in~$(0,1]$, in line with the probabilistic intuition common in WMC, the upper-bound reasoning developed later applies over the entire positive range. For numerical stability, one may equivalently work in the log domain by replacing products with sums. This is often preferable in floating-point implementations, as repeated multiplication of many positive weights can easily cause underflow, overflow, or loss of precision.
\end{remark}

\subsection{Comparison to Existing Approaches}


Table~\ref{tab:comparison} positions Weighted Model Enumeration with respect to the main SAT-based paradigms it relates to, namely AllSAT, Weighted Model Counting, and MaxSAT. While these approaches share individual aspects with WME (such as model enumeration, weight handling, or optimization objectives), none of them combine these features within a single solver framework.

AllSAT \cite{toda2016implementing, liang2022allsatcc, spallitta2025disjoint} provides complete model enumeration but lacks any notion of weights or ranking, making selective exploration infeasible without exhaustive enumeration. WMC \cite{dudek2019addmc, dudek2020dpmc, dilkas2021conditionalwmc} supports weights but collapses all models into a single aggregate quantity, offering no access to individual solutions. MaxSAT \cite{martins2014openwbo, nadel2024ttopenwboinc} optimizes for a single best model and does not support ranked or threshold-based enumeration without repeated solver invocations.

In contrast, WME integrates model enumeration, weight-aware search, and weight-based pruning directly into a CDCL-style solver. This enables native support for top-$k$ and threshold-based enumeration with incremental output, allowing entire regions of the search space to be pruned based on weight bounds when needed.

\vspace{-0.5em}
\begin{table}[t!]
\centering
\begin{tabular}{@{}lcccc@{}}
\toprule
\textbf{Feature}             & \textbf{AllSAT} & \textbf{WMC} & \textbf{MaxSAT} & \textbf{WME} \\\midrule
Model enumeration            & \checkmark      & \texttimes   & \texttimes      & \checkmark          \\
Weight support               & \texttimes      & \checkmark   & \checkmark      & \checkmark          \\
Top-$k$/threshold enumeration& \texttimes      & \texttimes   & $\sim$\checkmark\footnotemark & \checkmark \\
\bottomrule
\end{tabular}
\caption{Comparison of WME against similar SAT-based paradigms.}
\label{tab:comparison}
\end{table}
\footnotetext{Some MaxSAT extensions can retrieve multiple near-optimal models via repeated invocations, but this remains an indirect use of MaxSAT rather than a native ranked-enumeration operation.}

\section{Integrating Weight Reasoning into CDCL-Based Frameworks}

Weighted Model Enumeration extends standard CDCL-based enumeration by incorporating Boolean search with weight-based reasoning and threshold constraints, while preserving the overall structure of the CDCL loop. In particular, Boolean constraint propagation and Boolean clause learning remain unchanged.


The key extension is the integration of 
\emph{weight-based conflict analysis} into the CDCL loop. In addition to detecting logical conflicts, the solver continuously evaluates whether the current partial assignment can still satisfy an active weight bound~$\theta$, which may represent either a fixed threshold or a dynamically updated cutoff (e.g., the current $k$-th best score in top-$k$ enumeration). If the maximum achievable weight of a branch falls below~$\theta$, the branch is declared infeasible, and a weight conflict is analyzed, resulting in a learned clause that prunes the corresponding region of the search space. Conceptually, this mechanism plays a role analogous to theory-conflict detection in lazy SMT solvers \cite{barrettSatisfiabilityModuloTheories2021}, but operates over simple weight bounds rather than complex theory constraints.

\subsection{Weight Propagation and Pruning}
\label{subsec:weightprop}

WME extends CDCL search with a numerical dimension: besides logical consistency, the solver must track whether the current partial assignment can still be extended to a model whose weight reaches the active threshold~$\theta$. The goal is to detect as early as possible when this is no longer feasible. This process, called \emph{weight-based pruning}, complements standard Boolean propagation and provides the condition that triggers weight conflicts. Each literal~$\ell$ is associated with a real-valued weight~$w(\ell)$. For every Boolean variable~$A_i$, we define:
\[
\mathit{best}(A_i) = \max\{w(A_i), w(\neg A_i)\}
\]
For each partial assignment~$\mu$, the solver maintains two quantities: the partial weight~$w(\mu)$ on the assigned variables, and an optimistic residual bound~$I_{\max}(\mu)$ on the unassigned variables:
\[
I_{\max}(\mu) \;:=\; \prod_{A \in \mathcal{V} \setminus \mathrm{vars}(\mu)} \mathit{best}(A).
\]
Algorithm~\ref{algo:weights} maintains these quantities incrementally. At initialization (lines~\ref{line:init-wp}--\ref{line:init-imax}), $w(\mu)$ is set to~1, and $I_{\max}(\mu)$ is the product of all $\mathit{best}(A_i)$, corresponding to the assumption that unassigned variables will take their highest-weight polarity. When assigning a literal~$\ell$ (lines~\ref{line:assignfun-start}--\ref{line:assignfun-end}), the solver multiplies $w(\mu)$ by~$w(\ell)$ and removes the corresponding $\mathit{best}(A_i)$ contribution from~$I_{\max}(\mu)$. Unassigning a literal (lines~\ref{line:unassignfun-start}--\ref{line:unassignfun-end}) reverses both updates. 
Both operations preserve the invariant
\[
\forall \eta \supseteq \mu : \quad w(\eta) \le w(\mu) \cdot I_{\max}(\mu).
\]
That is, $w(\mu) \cdot I_{\max}(\mu)$ is an upper bound on the weight achievable by any completion of~$\mu$. For brevity, we write $\mathit{best}(\ell)$ to denote $\mathit{best}(\mathrm{var}(\ell))$.

The pruning test (lines~\ref{line:weightconflict-start}--\ref{line:weightconflict-end}) is invoked after Boolean propagation. If $w(\mu) \cdot I_{\max}(\mu) < \theta$, even an ideal extension cannot reach the lower bound. In this case, we prune the branch, and a weight conflict can be reported\footnote{A symmetric test based on the pessimistic bound $I_{\min}$ could prune branches that exceed an upper bound, though, for clarity, only the lower-bound variant is shown here.}.

\begin{algorithm}[t]
\caption{{\sc Weight-State Maintenance and Pruning}}
\label{algo:weights}
\begin{algorithmic}[1]
\Require Current assignment $\mu$

\State $w(\mu) \gets 1$ \label{line:init-wp}  \Comment partial assignment weight
\State $I_{\max}(\mu) \gets \prod_{A \in \mathcal{V}} \mathit{best}(A)$
  \Comment optimistic residual \label{line:init-imax}

\Statex

\Function{Assign}{$\mu,\ell$} \label{line:assignfun-start}
    \State $\mu$.append($\ell$) 
    \State $w(\mu) \gets w(\mu) \cdot w(\ell)$
    \State $I_{\max}(\mu) \gets I_{\max}(\mu) / \mathit{best}(\ell)$
\EndFunction \label{line:assignfun-end}

\Statex

\Function{Unassign}{$\mu,\ell$} \label{line:unassignfun-start}
    \State $\mu$.remove($\ell$)
    \State $w(\mu) \gets w(\mu) / w(\ell)$
    \State $I_{\max}(\mu) \gets I_{\max}(\mu) \cdot \mathit{best}(\ell)$
\EndFunction \label{line:unassignfun-end}

\Statex

\Function{WeightConflict}{$\mu$} \label{line:weightconflict-start}
    \If{$w(\mu) \cdot I_{\max}(\mu) < \theta$} 
        \State \Return \textbf{true} \Comment prune branch
    \EndIf
    \State \Return \textbf{false}
\EndFunction \label{line:weightconflict-end}

\end{algorithmic}
\end{algorithm}

\begin{proposition}[Weight pruning soundness]
\label{prop:pruning-sound}
Let $\mu$ be any partial trail reachable by the main CDCL loop, and let $w(\mu)$ and $I_{\max}(\mu)$ be maintained as in Algorithm~\ref{algo:weights}.
Then:
\[
\forall \eta \supseteq \mu:\quad w(\eta) \le w(\mu) \cdot I_{\max}(\mu).
\]
In particular, if $\textsc{WeightConflict}(\mu)$ returns \textbf{true}, i.e., $w(\mu) \cdot I_{\max}(\mu) < \theta$, then no completion $\eta \supseteq \mu$ can satisfy $w(\eta) \ge \theta$.
\end{proposition}

\begin{proof}
We prove that $w(\mu) \cdot I_{\max}(\mu)$ is always an upper bound on the weight of any extension of the current trail~$\mu$.

Initially $\mu = \varepsilon$, $w(\mu) = 1$, and $I_{\max}(\mu) = \prod_{A\in \mathcal{V}}\mathit{best}(A)$, so the bound holds trivially. Whenever a literal $\ell$ of variable $A$ is assigned, Algorithm~\ref{algo:weights} updates
$w(\mu) \leftarrow w(\mu) \cdot w(\ell)$ and $I_{\max}(\mu) \leftarrow I_{\max}(\mu) / \mathit{best}(\ell)$, thus maintaining:
\[
w(\mu) = \prod_{\ell\in\mu} w(\ell)
\quad\text{and}\quad
I_{\max}(\mu) = \prod_{B\in \mathcal{V}\setminus \mathrm{vars}(\mu)} \mathit{best}(B).
\]
Let $\eta \supseteq \mu$ be any completion. For every unassigned variable $B$, the literal $\ell_B$ chosen in $\eta$ satisfies $w(\ell_B) \le \mathit{best}(B)$ by definition. Thus:
\[
w(\eta)
= \prod_{\ell\in\mu} w(\ell)\cdot \prod_{B\in \mathcal{V}\setminus \mathrm{vars}(\mu)} w(\ell_B)
\le \prod_{\ell\in\mu} w(\ell)\cdot \prod_{B\in \mathcal{V}\setminus \mathrm{vars}(\mu)} \mathit{best}(B)
= w(\mu) \cdot I_{\max}(\mu).
\]
Hence $w(\mu) \cdot I_{\max}(\mu)$ is a valid upper bound for all extensions of $\mu$, and if $w(\mu) \cdot I_{\max}(\mu) < \theta$ then $w(\eta) < \theta$ for every $\eta\supseteq\mu$, establishing the correctness of pruning.
\end{proof}

\begin{example}
    Consider the following weighted instance, where $\phi$ is a subformula  
    over the set of variables $B_1, \cdots B_n$:
\begin{center}
\begin{minipage}[c]{0.44\textwidth}
\[
F = (A_1 \lor \phi) \land A_2 .
\]
\end{minipage}
\hfill
\begin{minipage}[c]{0.52\textwidth}
\centering
\begin{tabular}{c|@{\quad}ccc}
\toprule
 & $A_1$ & $A_2$ & $B_i \in \phi$ \\
\midrule
$w(\ell)$        & 0.6 & 0.3 & $\beta_i \in (0,1)$ \\
$w(\neg \ell)$   & 0.4 & 0.7 & $1 - \beta_i$ \\
\bottomrule
\end{tabular}
\end{minipage}
\end{center}
Suppose we enforce a lower-bound threshold $\theta = 0.2$.
Initially,
\begin{align}
I_{\max}(\mu) &= \max(0.4,0.6)\ \cdot \max(0.7,0.3) \cdot 
     \underbrace{\prod_{i=1}^n best(B_i)}_{\alpha} 
  = 0.42\,\alpha. \label{eq:imax}
\end{align}

Notice that due to the weight definition, $\alpha$ is guaranteed to be in the range $(0,1)$. Now suppose the trail is $\mu = \{\neg A_1, A_2\}$. Then $w(\mu) = 0.4 \cdot 0.3 = 0.12$.
We divide $best(A_1)$ and $best(A_2)$
from $I_{\max}(\mu)$, leaving the residual:
\[
I_{\max}(\mu)= \frac{I_{\max}(\mu)}{best(A_1) \cdot best(A_2)} = \frac{0.6 \cdot 0.7 \cdot \alpha}{0.6 \cdot 0.7} = \alpha
\].
Hence, for every extension of the partial assignment (including the optimistic completion):
\[
w(\mu) \cdot I_{\max}(\mu) \;=\; 0.12 \cdot \alpha \;\leq\; 0.12 \;<\; \theta \;=\; 0.2,
\]
Therefore, this branch can be pruned early, avoiding the assignment of any variables in $\phi$.
\end{example}



\subsection{Weight Conflict Analysis}
\label{subsec:weightconflict}

The previous section defines how a partial trail maintains an upper bound on
the achievable model weight and how branches that cannot reach the target
threshold~$\theta$ are pruned. When such pruning occurs, we interpret it as a
\emph{weight conflict}, analogous to a Boolean conflict in standard CDCL. 
The goal of weight-conflict analysis is to extract a short explanation from the current trail and learn a clause that prevents revisiting the same infeasible region.

\begin{definition}[Weight conflict]
Given a CNF formula $F$, a threshold~$\theta$, and a partial trail~$\mu$,
we say that $\mu$ induces a \textbf{weight conflict} if every total extension
$\eta \supseteq \mu$ fails to reach the threshold, i.e.,
\[
\forall \eta \supseteq \mu:\; w(\eta) < \theta.
\]
By Proposition~\ref{prop:pruning-sound}, $w(\mu) \cdot I_{\max}(\mu) <
\theta$ provides an efficiently checkable sufficient condition: if even the
optimistic upper bound on any extension of~$\mu$ is below~$\theta$, the branch
can be pruned.
\end{definition}

When a weight conflict is detected (line~\ref{line:weightconflict-start} of
Algorithm~\ref{algo:weights}), the solver triggers weight-conflict analysis (Algorithm \ref{algo:analyze-weight-conflict}).
The procedure first extracts a \emph{weight-conflict set} $S \subseteq \mu$ and
then learns the corresponding \emph{weight-conflict clause}
\[
C_w := \bigvee_{\ell \in S} \neg \ell,
\]
which blocks the simultaneous reoccurrence of all literals in~$S$.

\begin{algorithm}[t]
\caption{{\sc AnalyzeWeightConflict}$(\mu)$}
\label{algo:analyze-weight-conflict}
\begin{algorithmic}[1]
\Require Current trail $\mu$;\; threshold $\theta$
\State $S \gets$ \textsc{GreedyConflictSet}$(\mu, \theta)$
\State $C_w \gets \bigvee_{\ell \in S} \neg \ell$
\State \textsc{Learn}$(C_w)$ \Comment add learned weight-conflict clause
\State \textsc{Backjump}$(C_w)$ \Comment backjump per CDCL policy
\end{algorithmic}
\end{algorithm}

The extraction of the weight conflict set is shown in Algorithm~\ref{algo:greedyweight}. We use a greedy extraction strategy because it is easy to integrate, incurs little overhead, and provides a practical trade-off between clause compactness and extraction cost. Literals in the current trail~$\mu$ are sorted by ascending weight
$w(\ell)$ (line~\ref{line:gw-sort}). The procedure then greedily builds a set
$S \subseteq \mu$ (lines~\ref{line:gw-loop-start}--\ref{line:gw-loop-end}),
while maintaining ($i$) the partial weight
$w(S)=\prod_{\ell\in S} w(\ell)$ and ($ii$) an optimistic residual bound $I_{\max}(S)$ for
the variables not yet fixed by~$S$. 
As soon as $w(S) \cdot I_{\max}(S) < \theta$ (line~\ref{line:gw-threshold}), the
set $S$ is used to generate $C_w$. 

While stronger conflict-set extraction procedures could be considered, e.g., by minimizing the learned weight-conflict clause, we use a lightweight greedy strategy to keep the extraction cost negligible within the CDCL loop,  following prior work that prioritizes low-overhead reasoning during search \cite{spallitta2024disjoint}. More compact or cost-aware weight-conflict explanations are orthogonal to our framework and left as future work.

\begin{algorithm}[t]
\caption{{\sc GreedyConflictSet}$(\mu,\theta)$}
\label{algo:greedyweight}
\begin{algorithmic}[1]
\Require Current trail $\mu$;\; threshold $\theta$
\State $S \gets \emptyset$
\State $w(S) \gets 1$
\State $I_{\max}(S) \gets \prod_{A\in\mathcal V}\mathit{best}(A)$
\State $L \gets$ literals of $\mu$ sorted by ascending $w(\ell)$ \label{line:gw-sort}
\For{each $\ell \in L$} \label{line:gw-loop-start}
   \State $S \gets S \cup \{\ell\}$
   \State $w(S) \gets w(S) \cdot w(\ell)$
   \State $I_{\max}(S) \gets I_{\max}(S) / \mathit{best}(\ell)$
   \If{$w(S) \cdot I_{\max}(S) < \theta$} \label{line:gw-threshold}
      \State \Return $S$ \label{line:gw-return}
   \EndIf
\EndFor \label{line:gw-loop-end}
\State \Return $S$ \label{line:gw-full} \Comment{All literals in $\mu$ needed}
\end{algorithmic}
\end{algorithm}

\begin{proposition}[Soundness of greedy weight-conflict learning]
\label{prop:greedy-sound}
Let $\mu$ be a trail such that $\textsc{WeightConflict}(\mu)$ holds.
Let $S\subseteq \mu$ be the set returned by \textsc{GreedyConflictSet}$(\mu,\theta)$,
and let $C_w$ be the learned weight-conflict clause.
Assume that when the procedure returns $S$ it satisfies
$w(S) \cdot I_{\max}(S) < \theta$, where
$w(S)=\prod_{\ell\in S} w(\ell)$ and
$I_{\max}(S)=\prod_{A\in \mathcal{V}\setminus \mathrm{vars}(S)} \mathit{best}(A)$
are maintained as in Algorithm~\ref{algo:greedyweight}.
Then every total model $\eta$ of $F$ with $w(\eta)\ge \theta$ satisfies $C_w$.
Consequently, learning $C_w$ cannot eliminate any model
of weight at least~$\theta$.
\end{proposition}

\begin{proof}
Let $S$ be the set returned by \textsc{GreedyConflictSet}. By assumption, at
return time the procedure has $w(S) \cdot I_{\max}(S) < \theta$.

We first show that this implies that every completion of $S$ is below the
threshold. Let $\eta \supseteq S$ be any total assignment containing all literals in $S$. For each
variable $A \notin \mathrm{vars}(S)$, let $\ell_A$ be the literal on $A$ chosen
by $\eta$. By definition of $\mathit{best}$ we have $w(\ell_A) \le \mathit{best}(A)$,
hence
\[
w(\eta)
= \Big(\prod_{\ell\in S} w(\ell)\Big)\cdot
  \Big(\prod_{A\in \mathcal{V}\setminus \mathrm{vars}(S)} w(\ell_A)\Big)
\le w(S) \cdot
   \Big(\prod_{A\in \mathcal{V}\setminus \mathrm{vars}(S)} \mathit{best}(A)\Big)
= w(S) \cdot I_{\max}(S)
< \theta.
\]
Therefore:
\begin{equation}
\label{eq:proofer}
    \forall \eta \supseteq S:\quad w(\eta) < \theta.
\end{equation}

Now consider any total assignment $\eta$ of $F$ that falsifies $C_w=\bigvee_{\ell\in S}\neg\ell$.
Falsifying $C_w$ means $\ell$ is true in $\eta$ for all $\ell\in S$, hence
$S \subseteq \eta$. Applying Equation \ref{eq:proofer} returns $w(\eta) < \theta$.

Thus, no assignment of weight at least $\theta$ can falsify $C_w$.
Equivalently, every model $\eta$ with $w(\eta)\ge \theta$ satisfies $C_w$, so
learning $C_w$ cannot eliminate any model of weight at least~$\theta$.
\end{proof}

\begin{example}
Consider the following weighted CNF instance:
\begin{center}
\begin{minipage}[c]{0.44\textwidth}
\[
F = (A_1 \lor A_2 \lor \neg A_3) .
\]
\end{minipage}
\hfill
\begin{minipage}[c]{0.52\textwidth}
\centering
\begin{tabular}{c|@{\quad}ccc}
\toprule
 & $A_1$ & $A_2$ & $A_3$ \\
\midrule
$w(\ell)$        & 0.6 & 0.8 & 0.7 \\
$w(\neg \ell)$   & 0.4 & 0.2 & 0.3 \\
\bottomrule
\end{tabular}
\end{minipage}
\end{center}



Suppose the lower-bound threshold is $\theta=0.2$ and the current trail is $\mu = \{\neg A_1,\, A_2,\, \neg A_3\}$. The partial weight is $w(\mu) = 0.4 \cdot 0.8 \cdot 0.3 = 0.096$. Since $w(\mu) < \theta$, this branch is pruned. The greedy strategy identifies the literal with the lowest contribution (here, $\neg A_3$) as the main cause of the underweight condition (since $w(\neg A_3) \cdot best(A_1) \cdot best(A_2) = 0.144 < 0.2$), returning the weight conflict set $S = \{\neg A_3\}$. The corresponding weight conflict blocking clause $C_w = A_3$ is then learned.
\end{example}

\subsection{Residual-Aware Backtracking for Top-$k$ Enumeration}
\label{subsec:residualbacktrack}

In top-$k$ enumeration, whenever we find a satisfying assignment whose weight enters the current top-$k$, we tighten the active threshold~$\theta$ accordingly. After such an update, continuing the search from the last decision level is often pointless: the current trail~$\mu$ may already be so constrained that, even under the most optimistic completion of the remaining variables, no model strictly above~$\theta$ can be reached.

Algorithm~\ref{algo:residualbacktrack} uses the same optimistic bound from Algorithm~\ref{algo:weights} to determine how far to backtrack immediately after updating~$\theta$. Starting from the current trail~$\mu$, the procedure repeatedly removes the most recent literal, updating $w(\mu)$ and $I_{\max}(\mu)$ according to the incremental rules of Algorithm~\ref{algo:weights}. After each removal, the quantity $w(\mu)\cdot I_{\max}(\mu)$ remains an upper bound on the weight of any completion of the current trail. Any trail satisfying $w(\mu)\cdot I_{\max}(\mu)\le \theta$ cannot lead to an improved model, so we continue backtracking.

The procedure stops as soon as $w(\mu)\cdot I_{\max}(\mu)>\theta$, because some completion may still improve on the current cutoff. Thus, the solver can resume the search from the corresponding decision level by exploring the alternative branch. If the procedure backtracks to decision level~0 and the bound is still not above~$\theta$, then no model of weight greater than~$\theta$ exists, and top-$k$ enumeration can terminate.

\begin{algorithm}[t]
\caption{{\sc Residual-Aware Backtrack}$(\mu, \theta)$}
\label{algo:residualbacktrack}
\begin{algorithmic}[1]
\Require Current trail $\mu$;\; maintained $w(\mu)$, $I_{\max}(\mu)$;\; threshold $\theta$
\While{$\mu$ contains a literal at decision level $>0$} \label{line:rb-start}
    \State $\ell \gets \mu.\mathrm{pop}()$ \label{line:rb-pop}
    \State $w(\mu) \gets w(\mu) / w(\ell)$ \label{line:rb-updatew}
    \State $I_{\max}(\mu) \gets I_{\max}(\mu) \cdot \mathit{best}(\ell)$ \label{line:rb-updateimax}
    \If{$w(\mu) \cdot I_{\max}(\mu) > \theta$} \label{line:rb-check}
        \State \Return $\mu$ \label{line:rb-return} \Comment resume search from this trail and decision level
    \EndIf
\EndWhile
\State \Return $\mu$ \label{line:rb-root} \Comment return the decision level-0 trail and stop enumeration
\end{algorithmic}
\end{algorithm}

\begin{example}
Consider the following weighted CNF instance:
\begin{center}
\begin{minipage}[c]{0.44\textwidth}
\[
F = (A_1 \lor A_2) \land A_3 .
\]
\end{minipage}
\hfill
\begin{minipage}[c]{0.52\textwidth}
\centering
\begin{tabular}{c|@{\quad}ccc}
\toprule
 & $A_1$ & $A_2$ & $A_3$ \\
\midrule
$w(\ell)$        & 0.6 & 0.8 & 0.7 \\
$w(\neg \ell)$   & 0.4 & 0.2 & 0.3 \\
\bottomrule
\end{tabular}
\end{minipage}
\end{center}

The assignment $\eta=\{\neg A_1, A_2, A_3\}$ satisfies $F$ with $w({\eta}) = 0.4 \cdot 0.8 \cdot 0.7 = 0.224$. We assume the CDCL trail has the same order.
Assuming we are looking for the top-1 solution, the threshold $\theta$ is updated accordingly.

Removing $A_3$ with $\neg A_1, A_2$ fixed results in:
\[
\begin{array}{lcllcl}
(A_3=\top): & 0.4 \cdot 0.8 \cdot 0.7 = 0.224; &
(A_3=\bot): & 0.4 \cdot 0.8 \cdot 0.3 = 0.096.
\end{array}
\]
No extension of the partial assignment can be a new top-1 solution, so we can prune $A_3$.

Removing $A_2$ with only $\neg A_1$ fixed results in:
\[
\begin{array}{lcl lcl}
(A_2{=}\top, A_3{=}\top): & 0.4 \cdot 0.8 \cdot 0.7 = 0.224; &
(A_2{=}\bot, A_3{=}\top): & 0.4 \cdot 0.2 \cdot 0.7 = 0.056;\\
(A_2{=}\top, A_3{=}\bot): & 0.4 \cdot 0.8 \cdot 0.3 = 0.096; &
(A_2{=}\bot, A_3{=}\bot): & 0.4 \cdot 0.2 \cdot 0.3 = 0.024.
\end{array}
\]

None is above $\theta$, so $A_2$ can be safely removed. At this point, the remaining trail is $\mu=\{\neg A_1\}$. Its optimistic residual bound is $I_{\max}(\mu) = 0.6 \cdot 0.8 \cdot 0.7 = 0.336,$
which is $> \theta$, so there is a chance that an assignment with a higher score satisfying $F$ does exist. The solver backtracks and flips $A_1$. 
\end{example}



\subsection{Weight-based Variable Priority for Top-$k$ Enumeration}
\label{subsec:weight-priority}

Many real-world problems are not naturally given in CNF and are first encoded using structure-preserving encodings such as Tseitin~\cite{Tseitin1968} or Plaisted--Greenbaum~\cite{PlaistedGreenbaum1986}. The newly introduced variables in these encodings (and, more generally, any variable whose two polarities contribute the same weight) are often \emph{weight-irrelevant}: they do not affect model ranking under a weight-based threshold.

Let $\mathcal{V}$ be the full variable set. We partition $\mathcal{V}$ into \emph{weight-relevant} variables $W_r$ and \emph{weight-irrelevant} variables $W_i$.
Accordingly, any assignment $\eta$ can be decomposed as $\eta = \eta_r \cup \eta_i$ with $\eta_r$ over $W_r$ and $\eta_i$ over $W_i$. The contribution of $W_i$ is constant, denoted $\theta_i$, so we maintain a lower-bound threshold $\theta_r$
on the restricted score over $W_r$ only.
The branching heuristic prioritizes variables in $W_r$ so that weight-relevant variables are assigned first.

For any assignment $\eta_r$ over $W_r$, the assignment weight is defined such that:
\begin{equation*}
\forall\ \eta_i\ \text{over}\ W_i:\quad
w(\eta_r \wedge \eta_i) = w(\eta_r) \cdot \theta_i    
\end{equation*}    
Once all variables in $W_r$ have been assigned, if $w(\eta_r)<\theta_r$, then no extension over $W_i$ can reach the current threshold. During partial search over $W_r$, the same reasoning applies using an optimistic residual bound over the still-unassigned variables in $W_r$.

If all variables in $W_r$ are assigned and $w(\eta_r) \ge \theta_r$, the solver checks extendability to a total satisfying assignment including variables in $W_i$. 
When this succeeds, $\theta_r$ is updated accordingly, and residual-aware backtracking is applied with respect to the last decision on a variable in~$W_r$; assignments on $W_i$ made afterward cannot improve the score.
Conversely, if the optimistic bound over the remaining variables in $W_r$ drops below $\theta_r$, then variables in $W_i$ are skipped, and weight-based conflict analysis is triggered.

\begin{example}
Consider the following weighted CNF instance:
\begin{center}
\begin{minipage}[c]{0.44\textwidth}
\[
F = (A_1 \lor A_2)\ \land\ A_3\ \land\ (\neg Z_1 \lor A_1)\ \land\ (\neg Z_2 \lor A_2),
\]
\[
W_r = \{A_1, A_2, A_3\}, \qquad
W_i = \{Z_1, Z_2\}.
\]
\end{minipage}
\hfill
\begin{minipage}[c]{0.52\textwidth}
\centering
\begin{tabular}{c@{\quad}|ccccc}
\toprule
 & $A_1$ & $A_2$ & $A_3$ & $Z_1$ & $Z_2$ \\
\midrule
$w(\ell)$        & 0.9 & 0.8 & 0.7 & 0.5 & 0.5 \\
$w(\neg \ell)$   & 0.1 & 0.2 & 0.3 & 0.5 & 0.5 \\
\bottomrule
\end{tabular}
\end{minipage}
\end{center}

Since both polarities of each $Z_j$ carry the same weight, every assignment on $W_i$ contributes the same factor $\theta_i = 0.25$, so $W_i$ is weight-irrelevant.
We maintain a lower-bound threshold $\theta_r$ over $W_r$
and prioritize decisions on $W_r$.
Assume $\theta_r = 0.1$.

\noindent\textbf{Case 1:}
Suppose the trail on $W_r$ is
$\eta_r = \{A_1,\, \neg A_2,\, A_3\}$.
Then
\[
w(\eta_r) = 0.9 \cdot 0.2 \cdot 0.7 = 0.126.
\]
Since $w(\eta_r)\ge \theta_r$, we check extendability to a total
model satisfying $F$ (e.g., $Z_1=Z_2=\bot$ works).
Any completion over $W_i$ preserves the relative ranking of two assignments over $W_r$,
so if $w(\mu_{1_r}) > w(\mu_{2_r})$ then
\[
w(\eta_1) = w(\eta_{1_r}) \cdot \theta_i
  > w(\eta_{2_r}) \cdot \theta_i = w(\eta_2).
\]
After validating satisfiability, $\theta_r$ is updated with the full score
and residual-aware backtracking resumes from the last decision in $W_r$;
assignments to $(Z_1,Z_2)$ are irrelevant for the score.

\noindent\textbf{Case 2:}
Suppose instead we reach
$\eta_r' = \{\neg A_1, \neg A_2, A_3\}$,
giving
\[
w(\eta_r') = 0.1 \cdot 0.2 \cdot 0.7 = 0.014.
\]
Since $w(\eta_r') < \theta_r$, the solver can skip assigning $Z_1,Z_2$
and trigger a weight conflict.
\end{example}

In the top-$k$ setting, two total assignments may share the same total weight such that: ($i$) they have the same assignment on $W_r$, and ($ii$) different assignments on $W_i$. To avoid pruning such tie-equivalent solutions too early, we activate the optimization only after the first $k$ satisfying assignments have been collected. 
From that point onward, only assignments whose best achievable score is strictly better than the current $k$-th score can affect the maintained top-$k$ set.


\section{Chronological vs. Non-Chronological Backtracking}

Integrating weight reasoning into CDCL-based enumeration gives rise to two distinct solver designs, depending on how conflicts (both Boolean and weight-based) are resolved. The first adopts the classical non-chronological
clause-learning paradigm of modern SAT solvers \cite{toda2016implementing, liang2022allsatcc}.
The second combines CDCL with \emph{chronological backtracking} as used in recent AllSAT engines \cite{toda2016implementing, spallitta2025disjoint}. 
These two strategies have been explored in the non-weighted setting
\cite{spallitta2024disjoint, spallitta2025disjoint}, but weight-based pruning could accentuate their trade-offs in the weighted case. We emphasize that the distinction between chronological and non-chronological backtracking is well established in prior work \cite{toda2016implementing, mohle2019combining, spallitta2024disjoint}. Rather than redefining these notions, we focus on the different implementation requirements that arise when they are instantiated in AllSAT-style enumeration, in particular with respect to model blocking and exhaustive traversal.

\subsection{Chronological Backtracking Framework}

In the chronological framework, conflicts trigger chronological backtracking, but the solvers still incorporate CDCL-style conflict analysis \cite{mohle2019combining, spallitta2024disjoint}. To maintain chronological traversal, analysis is carried out up to the last unique implication point (UIP), instead of performing the usual first-UIP backjumping \cite{spallitta2024disjoint}. Blocking previously explored models is \emph{implicit}: once a branch is closed, the chronological traversal ensures it will not be revisited \cite{mohle2019combining}. 

The main advantage of this approach is its low memory footprint and fast propagation. Nevertheless, the performance of this design depends heavily on decision order. Early poor branching choices can delay the discovery of high-weight models, reducing the effectiveness of pruning. 
Restarts would break the
implicit blocking guarantee, so they must be disabled
\cite{spallitta2024disjoint}, which further increases sensitivity to branching heuristics. In practice, this variant performs best when the solution space is dense, such that implicit blocking remains efficient.

\subsection{Non-Chronological Backtracking Framework}

In the non-chronological framework, the solver backjumps to the first UIP rather than the last, thereby returning shorter clauses and enabling more effective backjumping \cite{biere2009handbook}. Model blocking is also \emph{explicit}: whenever a complete model is found, a clause consisting of the negation of all its decision literals is added to exclude it in subsequent search. These blocking clauses persist across restarts, enabling the solver to reset the trail while retaining information about previously explored or infeasible regions.

This approach is more robust to decision order, since restarts allow the solver to dynamically reorganize the search. Even if early decisions move the search away from promising regions, accumulated blocking and weight-conflict clauses guide the search toward high-weight models over time. The trade-off is that propagation becomes increasingly expensive as clauses accumulate, and memory usage may quickly increase.


\section{Experimental Evaluation}

We empirically evaluate the proposed WME framework on both synthetic and
real-world benchmarks. Two tasks are considered:
\emph{top-$k$} and \emph{threshold-based} enumeration. All experiments were conducted on a high-performance computing cluster. Each
run was bound to a single core of an Intel Xeon E5-2650 v2
@ 2.60GHz with 25~GB RAM. We implemented two CDCL-based WME variants, available at \url{https://github.com/giuspek/WME}, atop the
same \texttt{CaDiCaL}~\cite{biere2024cadical} codebase\footnote{Weights are represented using arbitrary-precision floating-point numbers in the implementation.}:
\begin{itemize}
    \item \texttt{WME-NCB:} Weighted model enumeration with
    \emph{non-chronological backtracking};
    \item \texttt{WME-CB:} Weighted model enumeration with
    \emph{chronological backtracking}.
\end{itemize}

For the top-$k$ experiments, we also compare against MaxSAT-based
enumeration baselines, namely the top-$k$ versions of \texttt{RC2} and
\texttt{MaxHS} discussed in the 2020 competition. We use the latest available versions of both tools, with the configurations described for the MaxSAT top-$k$ setting. For \texttt{RC2}, we use its public \texttt{enumerate()} interface. For MaxHS, the top-$k$ variant was released for the 2020 MaxSAT Evaluation. 
However, the currently available distribution does not provide a documented
way to invoke this mode. In particular, the accompanying helper script does not
expose usable invocation information through its help interface, preventing a
reliable reproduction of the competition-specific wrapper. We therefore reconstructed the procedure according to the competition description: after each optimum is found, we add a blocking clause for the returned solution and re-optimize under the accumulated blocking clauses.
%

\subsection{Benchmarks}




We evaluate weighted model enumeration on five benchmark families covering random, weighted-counting,
and MaxSAT-derived instances. All instances are provided in the WME codebase.

The first family, \texttt{rnd3sat-1.5}, consists of 525 random 3-CNF formulae
with 25 to 45 variables and clause-to-variable ratio~$1.5$. This low-density regime is commonly used in AllSAT evaluations because it yields many satisfying assignments while remaining non-trivial for exhaustive enumeration~\cite{bayardo1997using,spallitta2025disjoint}. These instances
typically admit many satisfying assignments, making them suitable for evaluating
threshold-based enumeration and for isolating the effect of weight-based
pruning. The second family, \texttt{uf200-860}, contains 100 SATLIB random
3-CNF instances with 200 variables and ratio~$4.28$, a standard hard SAT regime
used here as a stress test for weighted pruning. Literal weights for both random
families are sampled from~$(0,1)$.

The third family, \texttt{bayes}, consists of 389 literal-weighted CNF instances
from~\cite{phan2022dpo}, derived from Bayesian networks
in~\cite{sang2005performing}. These benchmarks are commonly used for weighted
model counting and MPE inference, and therefore naturally induce a ranking of
satisfying assignments by weight. We also include instances from the 2020
Weighted Model Counting Competition, denoted \texttt{wmc-comp}. We use the 2020
set to match the year of the MaxSAT top-$k$ competition and to avoid later
benchmark families with more articulated backbone structure, which could
disproportionately affect the MaxSAT-based baselines.

Finally, \texttt{maxsat-comp} contains a subset of instances from the MaxSAT 2020
top-$k$ competition. Since WME assigns weights to literals, while MaxSAT
instances may contain arbitrary soft clauses, not all MaxSAT benchmarks are
directly convertible to our setting. We therefore retain only instances whose
soft clauses are all unit clauses. Among the 70 available top-$k$ instances,
39 satisfy this condition and are converted into literal-weighted WME instances.

\subsection{Top-$k$ Enumeration}


\begin{figure}[!t]
  \centering

  \begin{subfigure}[t]{0.44\linewidth}
    \centering
    \includegraphics[width=\linewidth]{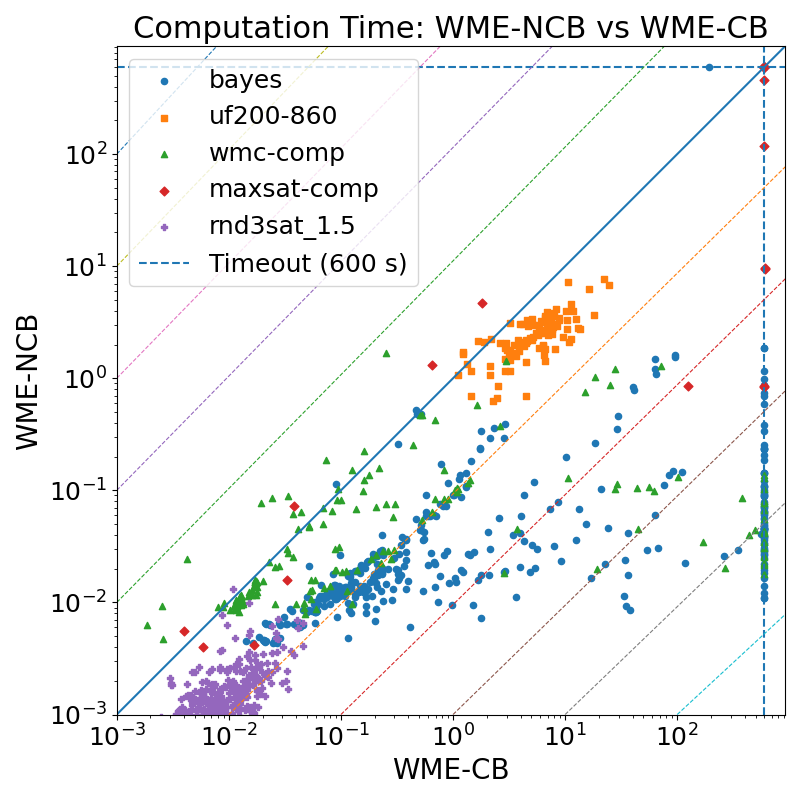}
    \caption{\texttt{WME-NCB} vs. {\tt WME-CB}}
    \label{fig:topk-dpo}
  \end{subfigure}
  \hfill
  \begin{subfigure}[t]{0.44\linewidth}
    \centering
    \includegraphics[width=\linewidth]{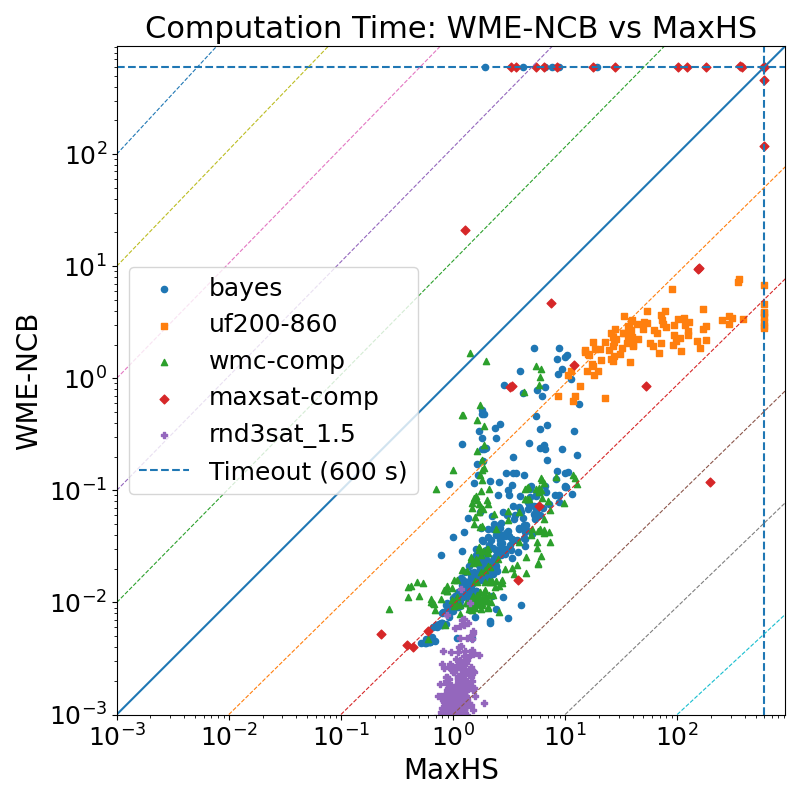}
    \caption{\texttt{WME-NCB} vs. {\tt MaxHS}}
    \label{fig:topk-maxhs}
  \end{subfigure}
  \hfill
  \begin{subfigure}[t]{0.44\linewidth}
    \centering
    \includegraphics[width=\linewidth]{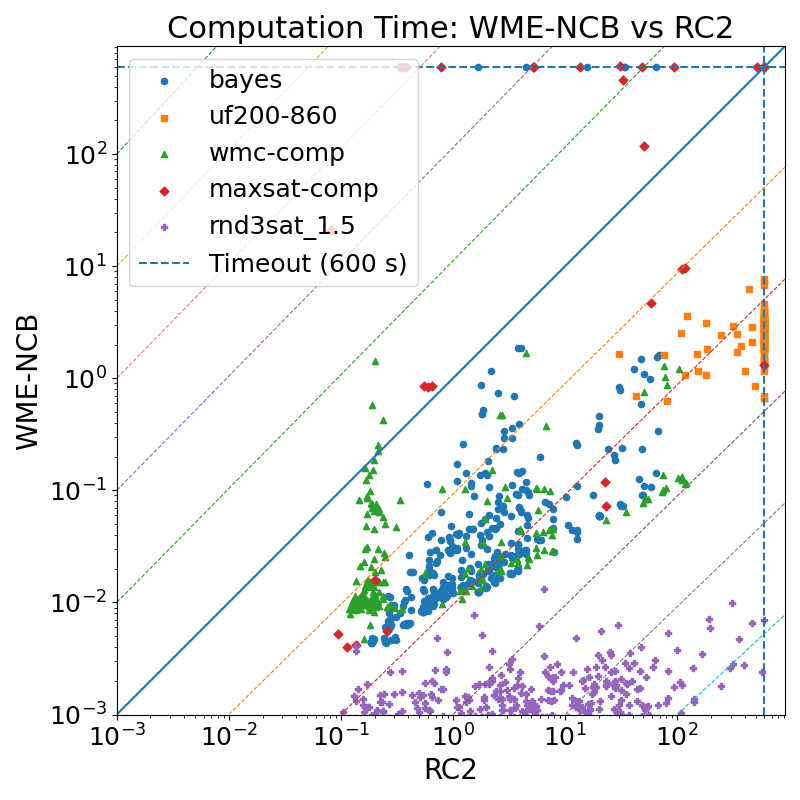}
    \caption{\texttt{WME-NCB} vs.\ {\tt RC2}}
    \label{fig:topk-overhead}
  \end{subfigure}

  \caption{Top-$50$ enumeration experiments.}
  \label{fig:topk-experiments}
\end{figure}

We first evaluate the \emph{top-$50$} task, where the goal is to enumerate the $50$ highest-weight models within the time limit (top-50 was one of the settings considered in the MaxSAT top-$k$ competition).

Figure~\ref{fig:topk-experiments} compares \texttt{WME-NCB} against \texttt{WME-CB}, \texttt{MaxHS}, and \texttt{RC2}. Overall, \texttt{WME-NCB} is the most robust WME configuration. Compared with \texttt{WME-CB}, \texttt{WME-NCB} benefits from non-chronological backtracking, explicit blocking clauses, and restarts, which allow the solver to move more aggressively across the search space after promising regions have been explored. 

The comparison with \texttt{MaxHS} and \texttt{RC2} is more nuanced. On several benchmark families, especially \texttt{wmc-comp}, \texttt{bayes}, and \texttt{rnd3sat-1.5}, \texttt{WME-NCB} is competitive and often faster. This supports the main motivation of our approach: once multiple high-weight solutions must be produced, integrating weight reasoning directly inside the enumeration procedure can avoid the repeated optimization-and-blocking pattern used by MaxSAT-based approaches.

On the \texttt{maxsat-comp} family, however, the picture is mixed. Both \texttt{MaxHS} and \texttt{RC2} perform very well on many of these instances, whereas \texttt{WME-NCB} exhibits several timeouts. We attribute this behavior to the structure of these MaxSAT-derived benchmarks. In particular, many of the best solutions have very similar, and in some cases identical, weights. In this regime, the top-$50$ cutoff does not become sufficiently selective: once the current bound reaches a value shared by many assignments, weight-based pruning cannot discard large parts of the remaining search space without risking the loss of tied top-$k$ solutions. By contrast, specialized MaxSAT solvers can exploit their highly optimized core-guided and branch-and-bound machinery on these instances.

Table~\ref{tab:top50-extended} reports, for each benchmark family, the number of solved instances and the number of instances on which each solver is the fastest. The results further suggest that the effectiveness of top-$k$ weighted enumeration depends not only on the number of requested models, but also on the distribution of model weights: when the weight spectrum is well-separated, the cutoff quickly becomes informative; when many top solutions are tied or nearly tied, pruning is weaker.

\begin{table*}[]
\centering
\scriptsize
\resizebox{\textwidth}{!}{%
\begin{tabular}{|l|cc|cc|cc|cc|cc|}
\hline
Solver & \multicolumn{2}{c|}{\texttt{maxsat-comp} (of 39)} & \multicolumn{2}{c|}{\texttt{wmc-comp} (of 200)} & \multicolumn{2}{c|}{\texttt{bayes} (of 389)} & \multicolumn{2}{c|}{\texttt{uf200-860} (of 100)} & \multicolumn{2}{c|}{\texttt{rnd3sat-1.5} (of 525)} \\
\cline{2-11}
 & Solved & Faster & Solved & Faster & Solved & Faster & Solved & Faster & Solved & Faster \\
\hline
\texttt{MaxHS} & 31 & 4 & 200 & 0 & 389 & 5 & 92 & 0 & 525 & 0 \\
\texttt{RC2} & 30 & 10 & 200 & 4 & 389 & 1 & 23 & 0 & 525 & 0 \\
\texttt{WME-CB} & 19 & 14 & 181 & 19 & 310 & 3 & 100 & 6 & 525 & 1 \\
\texttt{WME-NCB} & 21 & 9 & 200 & 177 & 383 & 380 & 100 & 94 & 525 & 524 \\
\hline
\end{tabular}%
}
\caption{Top-50 enumeration statistics across benchmark families. Notice that 2 problems from {\tt maxsat-comp} were solved by no solver.}
\label{tab:top50-extended}
\end{table*}

\subsection{Threshold-Based Enumeration}

\begin{figure}[t]
\centering
\begin{minipage}[c]{0.45\textwidth}
  \centering
  \includegraphics[width=\linewidth]{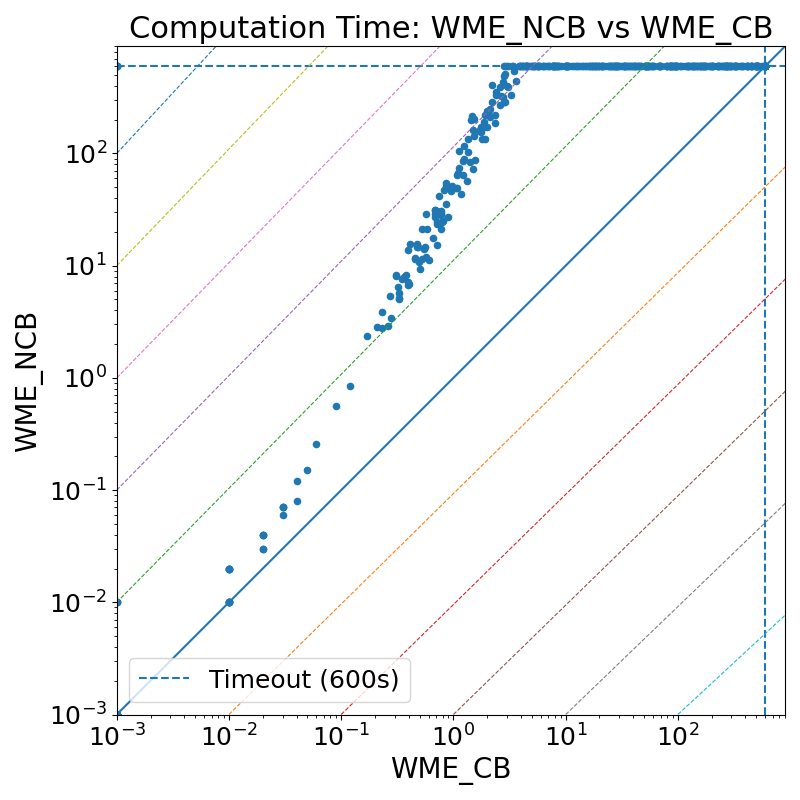}
\end{minipage}
\label{fig:thr-time}
\hfill
\begin{minipage}[c]{0.45\textwidth}
  \centering
  \begin{tabular}{l|c|c}
        \toprule
        Solver & Solved (out of 525) & PAR2 (s) \\
        \midrule
        \texttt{WME-NCB} & 356 & 457.5 \\
        \texttt{WME-CB}  & \textbf{450} & \textbf{267.5} \\
        \bottomrule
        \end{tabular}
  \captionof{table}{Timeout statistics for thresholded WME framework \\variants on \texttt{rnd3sat-1.5}.}
\end{minipage}
\caption{Threshold-based enumeration evaluation among the two WME frameworks on the rnd3sat-1.5.}
\label{fig:thr-time-vs-table}
\end{figure}

We next consider \emph{threshold-based} enumeration. We focus on the \texttt{rnd3sat-1.5} benchmarks, which feature a larger solution space and are standard for evaluating enumeration performance~\cite{bayardo1997using, spallitta2025disjoint}. For a formula with $n$ variables, we set the threshold $\theta$ to $0.5^n$, retaining a non-trivial fraction of solutions above it. Since this setting typically requires enumerating many models, we use a timeout of 600 seconds.

Figure~\ref{fig:thr-time-vs-table} reports the runtime comparison. Unlike in the top-$k$ case, \texttt{WME-CB} outperforms \texttt{WME-NCB}, especially on larger instances with many feasible models. This is expected, as non-chronological search accumulates blocking clauses that increasingly hinder propagation. These results are consistent with prior observations in unweighted enumeration~\cite{toda2016implementing, spallitta2025disjoint}.

\subsection{Ablation Studies}

\begin{figure}[!t]
  \centering
  \begin{subfigure}[t]{0.4\linewidth}
    \centering
    \includegraphics[width=\linewidth]{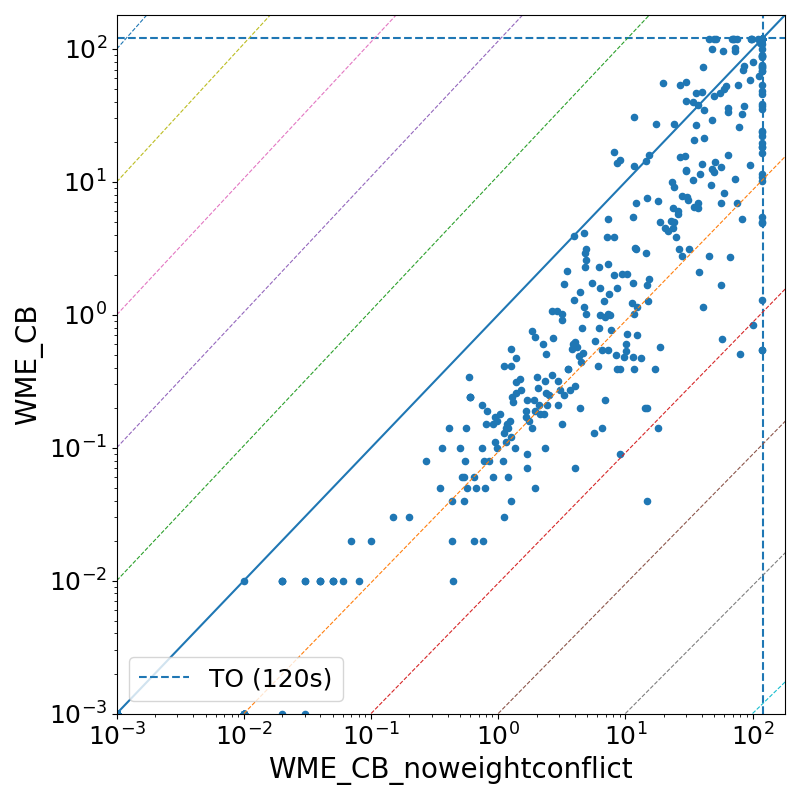}
    \label{fig:thr-ablation}
  \end{subfigure}
  \hfill
  \begin{subfigure}[t]{0.4\linewidth}
    \centering
    \includegraphics[width=\linewidth]{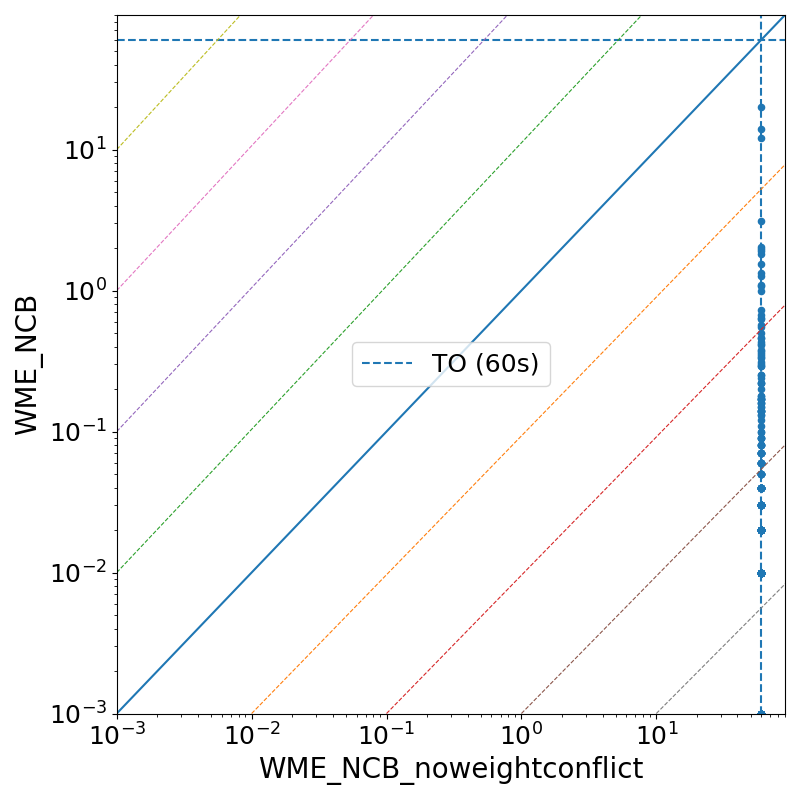}
    \label{fig:topk-ablation}
  \end{subfigure}
  \caption{Ablation studies. \textbf{Left:} threshold-based enumeration. \textbf{Right:} top-$50$ enumeration.}
  \label{fig:comparison_plots_ablation}
\end{figure}

Since WME integrates weight reasoning directly into the CDCL loop, it is natural to ask whether this integration is essential, or whether similar performance can be obtained by a pure AllSAT-style traversal followed by post hoc weight filtering. To this end, we compare our best-performing configuration with variants in which weight-based pruning is disabled.

For top-$50$ enumeration, we evaluate {\tt WME-NCB} on the challenging \texttt{bayes} benchmarks. For threshold-based enumeration, we consider {\tt WME-CB} on the \texttt{rnd3sat-1.5} instances. In both cases (Figure \ref{fig:comparison_plots_ablation}), enabling weight-based conflict analysis achieves a substantial speedup, confirming that early weight-driven pruning is essential under weight constraints. 

\section{Conclusion}

We introduced WME, a CDCL-based framework for
enumerating models of a formula $F$ under weight-based criteria. The framework combines weight-conflict analysis for early pruning with residual-aware backtracking, allowing the search to resume only from states that may still lead to improved solutions.

Our empirical results show that the choice between chronological and non-chronological backtracking is regime-dependent. In top-$k$ settings, where the bound typically tightens quickly, non-chronological backtracking with blocking clauses often reaches high-weight solutions faster. In threshold-based enumeration, where many feasible models may need to be generated, chronological backtracking avoids the additional propagation burden of blocking clauses and is often preferable. Overall, the most effective design depends on the interaction between bound tightening, feasible-region density, and propagation overhead.

One natural step is to make the solver adaptive, selecting between
chronological and non-chronological reasoning based on structural
properties of the instance and of the associated weight distribution. A second direction is to extend WME with projection and SMT reasoning, enabling weighted enumeration over variables of interest and under richer background theories. More broadly, these ideas may also be applied in future CDCL-based Weighted Model Integration algorithms to reduce the enumeration time \cite{spallitta2022smt, spallitta2024enhancing}.

\section*{Declaration on Generative AI}
 During the preparation of this work, the authors used Grammarly for grammar and spelling checks. After using this tool, the authors reviewed and edited the content as needed and took full responsibility for the publication’s content. 

\bibliography{sample-1col}

\appendix

\end{document}